\documentclass[submission,copyright,creativecommons]{eptcs}
\providecommand{\event}{V$^2$CPS 2016} 
\usepackage{breakurl}             

\usepackage{amsmath}
\usepackage{amssymb}
\usepackage[mathscr]{euscript}

\usepackage{subfig}

\usepackage{tikz}
\usepackage{pgfplots}

\title{Towards an Approximate Conformance Relation for \\ Hybrid I/O Automata}
\author{Morteza Mohaqeqi
\institute{Department of Information Technology\\
Uppsala University, Sweden}
\email{morteza.mohaqeqi@it.uu.se} 
\and
Mohammad Reza Mousavi
\institute{Centre for Research on Embedded Systems,\\
School of IT, Halmstad University, Sweden}
\email{m.r.mousavi@hh.se}
}
\def\titlerunning{Approximate Conformance for Hybrid I/O Automata}
\def\authorrunning{M.\ Mohaqeqi and M.R.\ Mousavi}

\newtheorem{definition}{Definition}

\newtheorem{example}[definition]{Example}

\newtheorem{theorem}[definition]{Theorem}

\newenvironment{proof}{Proof}{\hfill $\Box$}

\begin{document}
\maketitle

\begin{abstract}
Several notions of conformance have been proposed for checking the behavior of cyber-physical systems against their hybrid systems models.
In this paper, we explore the initial idea of a notion of approximate conformance that allows for comparison of both observable discrete actions and (sampled) continuous trajectories. 
As such, this notion will consolidate two earlier notions, namely the notion of Hybrid Input-Output Conformance (HIOCO) by M.\ van Osch and the notion of Hybrid Conformance by H.\ Abbas and G.E.\ Fainekos. We prove that our proposed notion of conformance satisfies a semi-transitivity property, which makes it suitable for a step-wise proof of conformance or refinement. 
\end{abstract}

\section{Introduction}\label{sec:into} 

Conformance testing is a practical and rigorous technique for verifying system correctness against a specification. 
In the context of cyber-physical systems, several notions of conformance have been proposed for testing such systems against their hybrid systems models, of which  hybrid input-output conformance (\textit{hioco}) \cite{vanOsch2006} and $(\tau, \epsilon)$-hybrid conformance (\textit{hconf}) \cite{Abbas14,AbbasThesis15} are two notable examples. 
As we noted earlier in  \cite{Mohaqeqi2014,Khakpour2015}, these two notions are substantially different in that hconf allows for approximate comparison of system dynamics, while hioco provides direct support for discrete actions and non-deterministic and partial specifications. Hence, as indicated in \cite{Mohaqeqi2014,Khakpour2015},  combining the advantages of these approaches will be beneficial and will lead to a flexible and  realistic definition of conformance with the above-mentioned characteristics.

In this paper, we define an approximate notion of conformance that allows for approximate comparison of both  (sampled) continuous signals and observable discrete actions. To this end, we take the notion of hconf \cite{Abbas14,AbbasThesis15}, and interpret discrete actions as special signals with discrete values ($0$ and $\infty$). As a first sanity check, we set out to prove that our notion of conformance is a pre-order (i.e., it is reflexive and transitive), making it suitable for step-wise conformance testing and refinement. Reflexivity immediately follows from the definition; transitivity, however, requires a twist in accumulating the conformance bounds (error margins). 
We formulate the notion of semi-transitivity and prove that our proposed notion is indeed semi-transitive. 

We illustrate the definitions using a simple example of a thermostat.  The thermostat has two input discrete actions ON and OFF and one output (continuous) variable $x$ denoting the current temperature. The dynamics of the temperature are different depending on the state of the thermostat, which is in turn determined by the discrete input actions. We provide a formalization of an ideal thermostat in hybrid (I/O) automata and analayze the conformance of its implementation with respect to the specification.

The rest of this paper is organized as follows.
In Section \ref{sec:rel}, we present an overview of the related work and position our approach with respect to the state of the art in the literature.  
In Section \ref{sec:pre}, we introduce the preliminary definitions regarding hybrid I/O automata and their traces. In Section \ref{sec:hioco}, we review the notion of hioco. 
In Section \ref{sec:ext}, we first review the notion of hconf and show how it can be extended to treat conformance for systems with discrete actions. Additionally, we state and prove it semi-transitivity property. In Section \ref{sec:conc}, we conclude the paper and present some directions for future research. 

\section{Related Work}\label{sec:rel}

There is a substantial literature on approximate notions of behavioral equivalence and/or pre-order for hybrid systems. We have already mentioned two examples of such notions, namely
hybrid input-output conformance (\textit{hioco}) \cite{vanOsch2006} and $(\tau, \epsilon)$-hybrid conformance (\textit{hconf}) \cite{Abbas14,AbbasThesis15}. Additionally, there are a number of notions of approximate equivalence / pre-order based on Metrics such as \cite{Girard2011,Julius2009,Tabuada07}. 
In \cite{Mohaqeqi2014,Khakpour2015}, we formally relate some of the basic notions of conformance for hybrid systems and their underlying semantic models. 

We expect the extension of our notion to the non-deterministic case to be straightforward; in \cite{vanOsch2006} and  \cite{AbbasThesis15} already non-deterministic models are considered. 
The extension to richer semantic domains, e.g., with probabilistics or stochastics, are far less trivial. We refer to \cite{Julius09,Zamani15} for related ideas in this direction.

\section{Preliminaries}\label{sec:pre} 

This section introduces the basic definitions and notations used throughout the rest of the paper, including a formal definition of hybrid I/O automata and their associated traces. 

\subsection{Notation and Basic Definitions}

We denote the set of real numbers by $\mathbb{R}$ and the set of non-negative integers by $\mathbb{N}$.
Given  a tuple $\alpha=(a_1,a_2,\ldots,a_n)$, the $i$th element of the tuple is denoted by $\gamma_i(\alpha)$, that is, $\gamma_i(\alpha)=a_i$. 

For a function $f$, we denote its domain by $\mathit{dom}(f)$. For a set $S$, the restriction of $f$ to $S$, denoted by $f \lceil S$, is a function with the domain $\mathit{dom}(f) \cap S$, where $(f \lceil S)(c) = f(c), \forall c \in S$.

In order to describe the system dynamics, we define and use the following notions of valuation and trajectory. 

\begin{definition}[Valuation]
Consider a set of real-valued variables $V$. A valuation of $V$ is a function of type $V \mapsto \mathbb{R}$, which assigns a real number to each variable $v \in V$. The set of all valuations of $V$ is denoted by $\mathit{Val}(V)$.
\end{definition}

\begin{definition}[Trajectory]
Let $D \subset \mathbb{R}$ be an interval. A trajectory $\sigma$ is a function $\sigma : D \rightarrow \mathit{Val}(V)$ that maps each element
in interval $D$ to a valuation. The set of all trajectories associated to $V$ is denoted by $\mathit{trajs}(V)$.
\end{definition}

Consider a set of variables $V$ and the corresponding set of valuations $\mathit{Val}(V)$. The \textit{restriction} of a valuation $\mathit{val} \in \mathit{Val}(V)$ to a set $V' \subseteq V$, denoted by $\mathit{val}\downarrow V'$, is a valuation $\mathit{val}' \in \mathit{Val}(V')$ where for all $v \in V'$, $\mathit{val}'(v)=\mathit{val}(v)$. Also, the restriction of a trajectory $\sigma : D \rightarrow \mathit{Val}(V)$ to a set $V' \subseteq V$, denoted by $\sigma \downarrow V'$, is a trajectory $D \rightarrow \mathit{Val}(V')$ where $(\sigma \downarrow V')(t)=\sigma(t) \downarrow V'$ for all $t \in D$.

For a trajectory $\sigma$, the first valuation of $\sigma$ (i.e., $\sigma(min(\mathit{dom}(\sigma)))$), if it exists, is denoted by $\sigma.\mathit{fval}$. Moreover, if $\mathit{dom}(\sigma)$ is right-closed, then  $\sigma.\mathit{lval}$ is defined as the last valuation of $\sigma$.

For an interval $D \subseteq \mathbb{R}$ and an arbitrary $ t \geq 0$, we define $D+t \triangleq \{t' + t | t' \in D\}$ as the time shift of $D$. Moreover, for a trajectory $\sigma$ with domain $D$, $\sigma+t$ is defined as a function with domain $D+t$ such that $\forall t' \in D: (\sigma+t)(t' + t)=\sigma(t')$.

For a trajectory $\sigma$, define $\tau \unrhd t \triangleq (\tau \lceil [t,\infty))-t$. Additionally, trajectory $\sigma$ is a prefix of trajectory $\sigma'$, denoted as $\sigma \leq \sigma'$, if and only if $\sigma=\sigma' \lceil \mathit{dom}(\sigma)$. Also, the concatenation of two trajectories  $\sigma$ and $\sigma'$ is defined by $\sigma ^\frown \sigma' \triangleq \sigma \cup \sigma'$.


\subsection{Hybrid Automata}

Hybrid I/O automata (HIOA) \cite{lynch2003hybrid} are defined as the extension of hybrid automata \cite{Alur1993} with an additional classification of external actions and variables as inputs and outputs. Hence, we first review the definition of hybrid automata and then introduce its extension to input and output actions and variables.

\begin{definition}[Hybrid Automata]
A hybrid automaton (HA) $\mathcal{A}=(W,X,Q,\Theta,E,H,D,\mathscr{T}$) consists of
\begin{itemize}
\item
A set $W$ of external variables and a set $X$ of internal variables, disjoint from each other. We write $V \equiv W \cup X$.

\item
A set $Q \subseteq \mathit{Val}(X)$ of states.

\item
A nonempty set $\Theta \subseteq Q$ of start states.

\item
A set $E$ of \textit{external actions} and a set $H$ of \textit{internal actions}, disjoint from each other. We write $A \equiv E \cup H$. 

\item 
A set $D \subseteq Q \times A \times Q$ of \textit{discrete transitions}. We use $\mathbf{x} \stackrel{a}{\rightarrow}_\mathcal{A} \mathbf{x'}$ as shorthand for $(\mathbf{x}, a, \mathbf{x'}) \in D$. We sometimes drop the subscript and  write $\mathbf{x} \stackrel{a}{\rightarrow} \mathbf{x'}$, when we think $\mathcal{A}$ should be clear from the context. We say that $a$ is enabled in $\mathbf{x}$ if there exists an $\mathbf{x'}$ such that $\mathbf{x} \stackrel{a}{\rightarrow} \mathbf{x'}$.

\item
A set $\mathscr{T}$ of trajectories for $V$ such that $\tau(t) \lceil X \in Q$ for every $\tau \in \mathscr{T}$ and $t \in \mathit{dom}(\tau)$. Given a trajectory $\tau \in \mathscr{T}$ we denote $\tau.\mathit{fval} \lceil X$ by $\tau.\mathit{fstate}$ and, if $\tau$ is closed, we denote $\tau.\mathit{lval} \lceil X$ by $\tau.\mathit{lstate}$. 
We require that the following axioms hold:

\textbf{T1} (\textit{Prefix closure})

\quad For every $\tau \in \mathscr{T}$ and every $\tau' \leq \tau$, $\tau' \in \mathscr{T}$.

\textbf{T2} (Suffix closure)
\quad
For every $\tau \in \mathscr{T}$ and every $t \in \mathit{dom}(\tau)$, $\tau \unrhd t \in \mathscr{T}$

\textbf{T3} (Concatenation closure)
\quad
Let $\tau_0$, $\tau_1$, $\tau_2$, $\ldots$ be a sequence of trajectories in $\mathscr{T}$ such that, for each nonfinal index $i$, $\tau_i$ is closed and $\tau_i.\mathit{lstate} = \tau_{i+1}.\mathit{fstate}$. Then $\tau_0 ^\frown \tau_1 ^\frown \tau_2 \ldots \in \mathscr{T}$.
\end{itemize}
\end{definition}

\begin{example}
\label{ex:HA}
Figure \ref{subfig:HAofTHERM}, due to \cite{Aerts16}, shows the hybrid automaton of the thermostat described in Section \ref{sec:into} 
 with the set of variables of $W=\{y\}$, $X=\{l, x\}$, where $l$ is a two-valued variable (mode\_ON and mode\_OFF) determining the \textit{location} of the automaton, $x$ is a continuous real-valued variable, and $y$ is always equal to $x$, i.e., $y(t)=x(t), \forall t \geq 0$. Further, the set of actions is specified as $E=\{\mathit{ON}\, \mathit{OFF}\}$, and $H=\{\}$. Also, we have 
\begin{equation}
Q=\{(l,x) \mid l \in \{mode\_ON, mode\_OFF \}, \ 0 \leq x \leq 20  \} \nonumber
\end{equation} 
and  $\Theta = \{(\mathit{mode\_ON},5)\}$.
Trajectories generated by this HIOA are specified according to the differential equations given for each location. Figure \ref{subfig:TRAJofTHERM} depicts a set of trajectories for variable $x$ (or equivalently external variable y). 
\end{example}

\begin{figure}
\centering
\subfloat[Hybrid automaton of the thermostat
\label{subfig:HAofTHERM}]{%
\begin{tikzpicture}[fill=white!20, align=center,scale=1,->]
	\path (1.1,6) node(a) [circle,text width=2.5cm,draw] { };
	\path (0.0,6.3) node(d)[text width=1cm]  {
		\fontsize{9}{9} 
		\begin{eqnarray}
		&\mathrm{mode\_ON} & \nonumber \\ 
		&\dot{x}(t)=-x(t)+20 & \nonumber \\ 
		&x(t) \leq 20	&\nonumber
		\end{eqnarray} 
		};
	\path (6.8,6) node(b) [circle,text width=2.5cm,draw] { };
	\path (6.0,6.3) node(d)[text width=1cm]  {
		\fontsize{9}{9} 
		\begin{eqnarray}
		&\mathrm{mode\_OFF} & \nonumber \\ 
		&\dot{x}(t)=-x(t) & \nonumber \\ 
		&x(t) \geq 0	&\nonumber
		\end{eqnarray} 
		};
	\draw[thick,black,>=stealth] (b) .. controls +(150:2cm) and +(30:2cm) .. node[above] {\textit{ON} ($x(t) \leq 2$)}  (a);
	\draw[thick,black,>=stealth] (a) .. controls +(-30:2cm) and +(-150:2cm) .. node[below]{\textit{OFF} ($x(t) \geq 18$)} (b);
\end{tikzpicture}
}\quad
\subfloat[A sample of the continuous dynamics of the system\label{subfig:TRAJofTHERM}]{
\begin{tikzpicture}[scale=0.7]
\begin{axis}[grid=both,
          xmin=0,ymin=0,xmax=10,ymax=20,
          xlabel=$t$,
		  ylabel=$x$($t$)
          ]
\addplot[red,domain=0:2,samples=20]  {-15*exp(-x)+20};
\addplot[blue,domain=2:4.2,samples=20]  {18*exp(2-x)+0};
\addplot[red,domain=4.2:6.4,samples=20]  {-18*exp(4.2-x)+20};
\addplot[blue,domain=6.4:8.6,samples=20]  {18*exp(6.4-x)+0};
\addplot[red,domain=8.6:10,samples=10]  {-18*exp(8.6-x)+20};
\addplot[white](.7,16) node[red,above]  {$\tau_0$};
\addplot[white](3.3,10) node[blue,above]  {$\tau_1$};
\addplot[white](5.3,16) node[red,above]  {$\tau_2$};
\addplot[white](7.7,10) node[blue,above]  {$\tau_3$};
\addplot[white](9.2,15) node[red,above]  {$\tau_4$};
\end{axis}
\end{tikzpicture}
}
\caption{Thermostat example}\label{fig:thermos}
\end{figure}
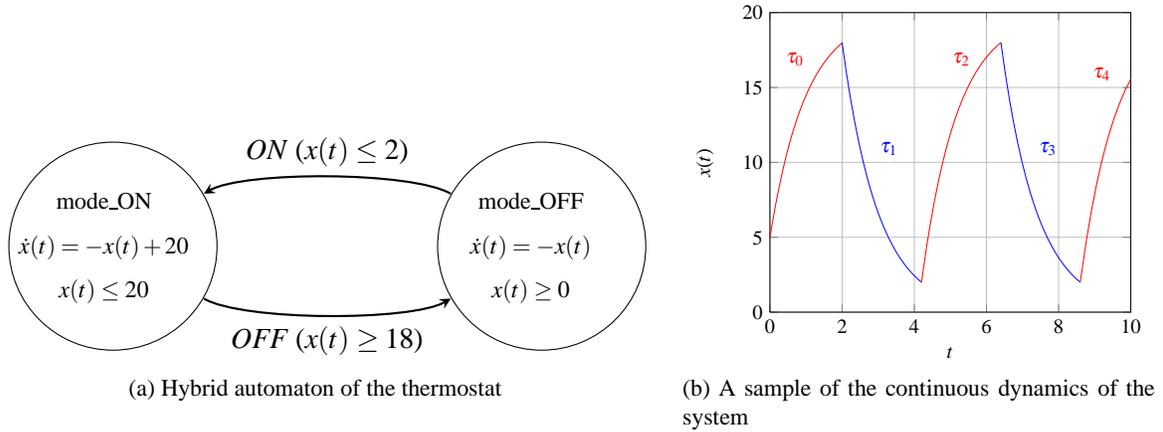

\subsection{Hybrid I/O Automata}

As mentioned before, an HIOA is a hybrid automaton in which the set of variables and actions are classified as input and output, as specified below.

\begin{definition}[Hybrid I/O Automata \cite{lynch2003hybrid}]
A hybrid I/O automaton (HIOA) is a tuple ($\mathcal{H},W_I,W_O,E_I,E_O$) where
\begin{itemize}
\item
$\mathcal{H}=(W,X,Q,\Theta,E,H,D,\mathscr{T})$ is a hybrid automaton.

\item
$W_I$ and $W_O$ partition $W$ into \textit{input} and \textit{output} variables, respectively. Variables in $Z\equiv X \cup W_O$ are called locally controlled.

\item
$E_I$ and $E_O$ partition $E$ into \textit{input} and \textit{output} actions, respectively. Actions in $L\equiv H \cup E_O$ are called locally controlled.

\item 
The following additional axioms are satisfied:
\begin{itemize}
\item
\textbf{E1} (Input action enabling)

For every  $\mathbf{x} \in Q$ and every $a \in E_I$, there exists $\mathbf{x'} \in Q$ such that $\mathbf{x} \stackrel{a}{\rightarrow} \mathbf{x'}$.

\item
\textbf{E2} (Input trajectory enabling)

For every  $\mathbf{x} \in Q$ and every $v \in \mathit{trajs}({W_I})$, there exists $\tau \in \mathscr{T}$ such that $\tau.\mathit{fstate} = \mathbf{x}$, 
$\tau \downarrow {W_I} \leq v$ and either
\begin{enumerate}
\item
$\tau \downarrow {W_I} = v$, or

\item
$\tau$ is closed and some $l \in L$ is enabled in $\tau.\mathit{lstate}$.

\end{enumerate}

\end{itemize}
\end{itemize}
\end{definition}

\begin{example}
\label{ex:HIOA}
Consider the hybrid automaton described in Example \ref{ex:HA}. Classifying the set of variables into input and output variables $W_I=\{\}$, $W_O=\{y\}$ and the set of actions into input and output actions as $E_I=\{\mathit{ON}\}$, $E_O=\{\mathit{OFF}\}$ reveals the corresponding HIOA.
\end{example}

{For an HIOA, a state is said to be \textit{agile} if it affords some trajectory, i.e., there exists trajectory $\sigma \in \mathscr{T}$ for which $s \stackrel{\sigma}{\rightarrow} $. To show that a state is agile, we add a spacial action $\xi$ to the set of output actions $E_O$ (the usage of this action is explained in the next section).}

\subsection{Hybrid sequences, executions, traces and solution pairs}
To describe the behavior of a hybrid (I/O) automaton, we use the notions of execution, trace and solution pair. To this end, we first define a notation of hybrid sequence.

\begin{definition}[Hybrid Sequence \cite{lynch2003hybrid}]
Take a set of variables $V$ and a set of actions $A$. Then, an ($V,A$)-sequence is defined as a finite or infinite sequence $\alpha=\tau_0,a_1,\tau_1,a_2,\ldots$, in which
\begin{itemize}
\item
each $\tau_i$ denotes a trajectory in $\mathit{trajs(V)}$;
\item
each $\alpha_i$ denotes an action in $A$;
\item
if $\alpha$ is finite then it ends with a trajectory;
\item
if $\tau_i$ is not the last trajectory in the sequence then  $\mathit{dom}(\tau_i)$  is closed.
\end{itemize}

Any ($V,A$)-sequence is called a hybrid sequence.
\end{definition}

Next, we define the notions of execution fragment and execution for a hybrid I/O automaton. 

\begin{definition}[Execution Fragment and Execution \cite{lynch2003hybrid}]
A hybrid sequence $\alpha=\tau_0,a_1,\tau_1,a_2,\ldots$ is an execution fragment for an HIOA $A$ if (i) each $\tau_i$ is a trajectory in $\mathscr{T}$, and (ii) if $\tau_i$ is not the last trajectory then $\tau_i.\mathit{lstate} \stackrel{a_{i+1}}{\rightarrow} \tau_{i+1}.\mathit{fstate}$. 
For an execution fragment $\alpha$, we define $\alpha.\mathit{fstate}$ to be $\tau_{0}.\mathit{fstate}$.

An execution fragment $\alpha=\tau_0,a_1,\tau_1,a_2,\ldots$ is called an \emph{execution} of a hybrid automata $\mathcal{A}$  if 
$\tau_0.\mathit{fstate}$ is a start state, i.e., $\tau_0.\mathit{fstate} \in \Theta$.

We write $\mathit{execs}_\mathcal{A}$ to denote the set of all executions of $\mathcal{A}$.
\end{definition} 

To capture the \textit{external} behavior of a hybrid (I/O) automaton, we use the notion of \textit{trace}. Intuitively, a trace of an execution specifies the sequence of its external actions and the evolution of the external variables.
\begin{definition}[Trace \cite{lynch2003hybrid}]
\label{def:trace}
Consider an execution $\alpha$. The trace of $\alpha$, denoted by $\mathit{trace}(\alpha)$, is the ($E,W$)-restriction of $\alpha$ (It is recalled that $E$ and $W$ denote the set of external actions and external variables, respectively).
\end{definition}

\begin{example}
\label{ex:trace}
Figure \ref{subfig:TRAJofTHERM} depicts a set of trajectories for variable $x$ of the HIOA specified in Example \ref{ex:HIOA}. According to this sample behavior, $\alpha=\tau_0,OFF,\tau_1,ON,\tau_2$ specifies a trace of the HIOA.
\end{example}

\begin{definition}[Solution Pair \cite{AbbasThesis15}]
\label{def:sol}
Assume that $u$ and $y$ are two traces for a HIOA $\mathscr{H}$; ($u$, $y$) is a solution pair to $\mathscr{H}$ if 
\begin{itemize}
\item
$\mathrm{dom}(u)=\mathrm{dom}(y)$, and 
\item
there exists a trace $\alpha$ to $\mathscr{H}$ such that $\mathrm{dom}(\alpha)=\mathrm{dom}(u)$, $u=\alpha \downarrow W_I$, and $y=\alpha \downarrow W_O$. 
\end{itemize}
\end{definition}

\section{I/O Conformance for Hybrid I/O Automata}\label{sec:hioco}
In this section, we review the notion of hioco \cite{vanOsch2006}, which is an exact notion of conformance, and adopt it for the HIOA formalism.
Subsequently, in the next section, we will adapt the ideas explored in this section to 
come up with an approximate notion of conformance.  
To this end, we first define the output and the trajectories associated with a state of an HIOA. 

\begin{definition}
Assume that  $\mathcal{A}=$ ($\mathcal{H},W_I,W_O,E_I,E_O$) is a hybrid IO automaton with the sets $Q$ and $\Theta$ of states and starting states, respectively. For a state $s \in \Theta$ we define
\begin{eqnarray}
\mathbf{out}(s)=\left\lbrace \begin{array}{ll}
                  \{o \in E_O \mid \{\xi\}\cup s\stackrel{o}{\rightarrow}\} 
                  			& \mathrm{, if \ } s \mathrm{\ is \ agile} \\ 
                  \{o \in E_O \mid s \stackrel{o}{\rightarrow} \} 
                  			& \mathrm{, otherwise}
                \end{array}
              \right.
\end{eqnarray}
Furthermore, for a set $C \subseteq Q$ we define:
\begin{equation}
\mathbf{out}(C) = \bigcup_{s \in C}  \mathbf{out}(s). \nonumber
\end{equation}
\end{definition}

\begin{definition}[$\mathbf{after}$ operator]
For an HIOA $\mathcal{A}$ and a trace $t$, $\mathbf{after}$ operator is defined as
\begin{equation}
\mathcal{A} \mathbf{after}\ t = \{ s \mid s_0 \stackrel{t}{\rightarrow}  s \}
\end{equation}
\end{definition}

\begin{definition}[Trajectories of a state]
For an HIOA $\mathcal{A}$ and a state $s$, $\mathbf{traj}(s)$ is defined as
\begin{equation}
\mathbf{traj} (s) = \{ \sigma \in \mathscr{T} \mid s \stackrel{\sigma}{\rightarrow}  \}
\end{equation}
\end{definition}

In order to define conformance, we need to focus on those  trajectories of the system under test for which a trajectory with the same behavior exists in the specification. 
This is achieved by the following infilter operator. 

\begin{definition}[$\mathbf{infilter}$]
Consider two sets of trajectories $\Sigma_I$ and $\Sigma_S$ on a set of variables $V$. Assume $V_I \subseteq V$ is the set of input variables; then
\begin{equation}
\mathbf{infilter} (\Sigma_I, \Sigma_S) = \{ \sigma \in \Sigma_I \mid \exists_{\sigma' \in \Sigma_S} : \sigma \downarrow V_I = \sigma' \downarrow V_I  \}
\end{equation}
\end{definition}

Finally, the notion of HIOCO is defined HIOA below. 

\begin{definition}[extension of $\mathbf{hioco}$\label{def:hioco}]
An (input-enabled) HIOA $I$ is said to hybrid input-output conform to another HIOA  $S$, denoted by $I\ \mathbf{hioco}\ S$, if and only if for all traces $t \in \mathit{traces}(S)$:
\begin{eqnarray}
& & \mathbf{out}(I\ \mathbf{after}\ t) \subseteq \mathbf{out}(S\ \mathbf{after}\ t), \mathrm{and} \nonumber \\ 
& & \mathbf{infilter} \left(\mathbf{traj}(I\ \mathbf{after}\ t),\mathbf{traj}(S\ \mathbf{after}\ t)\right) \subseteq \mathbf{traj}(S\ \mathbf{after}\ t)\nonumber
\end{eqnarray}
\end{definition}

\begin{example}
\label{ex:hioco}
Assume that we refer to the HIOA specified in Example \ref{ex:HIOA} by $S$. As specified in Example \ref{ex:trace}, $\alpha=\tau_0,\textit{OFF},\tau_1,\textit{ON},\tau_2$ is a trace of $S$. Then, we have $\mathbf{out}(S\ \mathbf{after}\ \alpha) = \{\textit{OFF} \}$. 
Now, assume an imaginary HIOA $I$ which generates a trace $\alpha'=\tau_0,\textit{OFF},\tau_1,\textit{ON},\tau_2,\textit{ON}$. 
It is seen that $\textit{ON} \in \mathbf{out}(I\ \mathbf{after}\ \alpha)$. As a consequence, $I$  does not conform to $S$ since $\mathbf{out}(I\ \mathbf{after}\ \alpha) \not\subseteq \mathbf{out}(S\ \mathbf{after}\ \alpha)$.
\end{example}

\section{Approximate Conformance for Hybrid I/O Automata}\label{sec:ext}

As it can be noted in Definition \ref{def:hioco}, 
HIOCO considers an implementation to be conforming when it exactly follows (features a non-empty subset of the behavior of) the specification. However, due to (unavoidable) inaccuracies during the implementation and testing of a CPS,  this is not a feasible and practical approach. We wish to allow for \textit{limited} deviations of the implementation from the specification. In this regard, Abbas et al. \cite{Abbas14,AbbasThesis15} define a notion of closeness which is measures the degree of similarity between the observed behavior of two systems. However, the approach ignores the input and output actions and mainly focuses on the continuous behavior of the system. In the remainder of this section, we first review the hybrid conformance relation proposed by Abbas et al.\ \cite{Abbas14,AbbasThesis15} and then extend it, following the recipe of HIOCO, with discrete input and output actions.

\subsection{Conformance Relation based on the Notion of Closeness}

In the hybrid conformance approach \cite{Abbas14,AbbasThesis15}, the behavior of a system is specified using a notion of trace in which the types of observable discrete actions are not explicitly specified.  Instead, the number of discrete jumps are considered relevant in the definition of hybrid conformance. This kind of trace, which we call an \textit{action-insensitive trace} (or an \textit{a-trace}) hereafter, is defined based on the notion of hybrid time domain, defined below.

\begin{definition}[Hybrid Time Domain \cite{AbbasThesis15}]
A hybrid time domain $E$ is a subset of $\mathbb{R}_{+} \times \mathbb{N}$ defined as
\begin{equation}
E=\bigcup_{j=0}^{J-1}[t_j,t_{j+1}]\times\{j \}
\end{equation}
where $0=t_0 \leq t_1 \leq t_2 \leq ... \leq t_J$. We denote the set of all hybrid time domains by $\mathbb{T}$. 
\end{definition}

The hybrid time domain is used to specify the evolution of a CPS regarding both continuous evolution (using time intervals $[t_j,t_{j+1}]$) and discrete jumps (using integer numbers $j$). Next, we define the notion of action-insensitive trace which is subsequently used to define the approximate notion of conformance.\footnote{In \cite{AbbasThesis15}, the term \textit{Hybrid Arc} is used to refer to this concept.}

\begin{definition}[Action-Insensitive Trace \cite{AbbasThesis15}]
Take a hybrid time domain $E$ and a set of variables $V$. An action-insensitive trace (a-trace) over $E$ is a function $\phi: E \rightarrow \mathit{Val}(V)$, where  for every $j$, $t \mapsto \phi(t,j)$ is absolutely continuous in $t$ over the interval $I_j=\{t | (t,j) \in E \}$. The set of all a-traces defined over the variable set $V$ is denoted by a-trace$(V)$. 
\end{definition}

As it can be noted, the notion of a-trace is an abstraction of a trace, as defined in Definition \ref{def:trace}. In fact, an a-trace can be obtained from a given trace by removing its actions and keeping only the number of jumps.

\begin{definition}[($\tau$,$\epsilon$)-closeness \cite{AbbasThesis15}]
\label{def:closeness}
Consider a test duration $T \in \mathbb{R}_{+}$, a maximum number of jumps $J \in \mathbb{N}$, and $\tau,\epsilon > 0$; then two a-traces $y_1$ and $y_2$ are said to be  ($\tau$,$\epsilon$)-close, denoted by $y_1 \approx_{(\tau,\epsilon)} y_2$, if 
\begin{enumerate}
	\item
	for all $(t,i) \in \mathrm{dom}(y_1)$ with $t\leq T, i \leq J$, there exists $(s,j) \in \mathrm{dom}(y_2)$  such that  $|t-s| \leq \tau$ and $\|y_1(t,i) -y_2(s,j)\| \leq \epsilon$, and

	\item
	for all $(t,i) \in \mathrm{dom}(y_2)$ with $t\leq T, i \leq J$, there exists $(s,j) \in \mathrm{dom}(y_1)$  such that  $|t-s| \leq \tau$ and $\|y_2(t, i) -y_1(s, j)\| \leq \epsilon$.
\end{enumerate}
\end{definition}

\begin{definition}[Conformance Relation \cite{AbbasThesis15}]
\label{def:conformance}
Consider two hybrid I/O automata $\mathscr{H}_1$ and $\mathscr{H}_2$. Given a test duration $T \in \mathbb{R}_{+}$, a maximum number of jumps $J \in \mathbb{N}$, and $\tau,\epsilon > 0$, \emph{$\mathscr{H}_2$ conforms to $\mathscr{H}_1$}, denoted by $\mathscr{H}_2 \approx_{(\tau,\epsilon)} \mathscr{H}_1$, if and only if for all solution pairs $(u,y_1)$ of $\mathscr{H}_1$,
 there exists a solution pair $(u, y_2)$ of $\mathscr{H}_2$ such that the corresponding output a-traces $y_1$ and $y_2$ are ($\tau$,$\epsilon$)-close.
\end{definition}

{Figure \ref{fig:hconf} shows two a-traces $y_1$ and $y_2$ where  
$y_1 \approx_{(0.8, 1)} y_2$. As mentioned, the notion of hconf ($\approx_{(\tau, \epsilon)}$) is not sensitive to the existence/absence of actions. As a result, if the HIOA generating $y_1$ produces  output action \textit{OFF} at $t=2$ and the other HIOA produces no action, the two a-traces are still regarded as conforming according to hconf.}

\begin{figure}
\centering
\begin{tikzpicture}[scale=0.7]
\begin{axis}[grid=both,
          xmin=0,ymin=0,xmax=10,ymax=20,
          xlabel=$t$,
		  ylabel=$y$($t$)
          ]
\addplot[domain=0:2,samples=20]  {-15*exp(-x)+20};
\addplot[domain=0.4:2.5,samples=20,dashed]  {-15*exp(0.5-x)+20-1};
\addlegendentry{$y_1$}
\addlegendentry{$y_2$}

\addplot[domain=2:4.2,samples=20]  {18*exp(2-x)+0};
\addplot[domain=4.2:6.4,samples=20]  {-18*exp(4.2-x)+20};
\addplot[domain=6.4:8.6,samples=20]  {18*exp(6.4-x)+0};
\addplot[domain=8.6:10,samples=10]  {-18*exp(8.6-x)+20};

\addplot[domain=2.5:4.7,samples=20,dashed]  {18*exp(2.5-x)+0-1};
\addplot[domain=4.7:6.9,samples=20,dashed]  {-18*exp(4.7-x)+20-1};
\addplot[domain=6.9:9.1,samples=20,dashed]  {18*exp(6.9-x)+0-1};
\addplot[domain=9.1:10,samples=10,dashed]  {-18*exp(9.1-x)+20-1};

\addplot[domain=2:3.5,samples=10,dotted]  {18};
\addplot[domain=2.5:3.5,samples=10,dotted]  {17};
\addplot[white](3.8,17.5) node[black]  {$1$};

\addplot[white](5.5,18.7) node[black]  {$0.5$};

\addplot[white](1,16) node[black]  {$\tau_1$};
\addplot[white](3,5) node[black]  {$\tau_2$};

\addplot[white](0.9,5) node[black]  {$\tau'_1$};
\addplot[white](3.3,10) node[black]  {$\tau'_2$};
\end{axis}

\draw[black,->] (2.3,4.6) -- (2.3,4.8); 
\draw[black,->] (2.3,5.35) -- (2.3,5.15); 

\draw[black,->] (4.15,5.35) -- (4.35,5.35); 
\draw[black,->] (5,5.35) -- (4.8,5.35); 

\draw[black,dotted] (4.75,4.8) -- (4.75,5.6); 
\draw[black,dotted] (4.4,5) -- (4.4,5.6); 

\end{tikzpicture}
\caption{Two signals which are $\approx_{\tau,\epsilon}$ for $\tau=0.8$ and $\epsilon=1$ but not for $\tau=0.8$ and $\epsilon=0.4$.}\label{fig:hconf}
\end{figure}
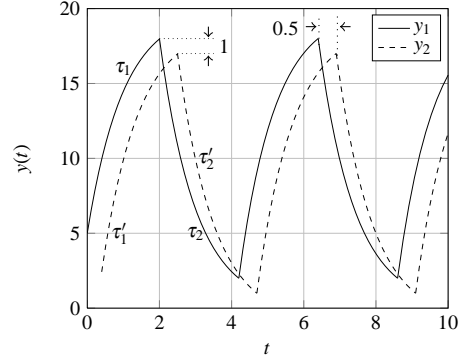

\subsection{Extension}

The goal of this section is to extend the notion of closeness and the corresponding conformance relation to the case of hybrid I/O automata. The main difference of our extended conformance with that of Abbas et al. \cite{Abbas14} is that 
we include the input and output actions for the comparison of two behaviors 
(i.e., the comparison of specification and implementation).



According to Definition \ref{def:closeness}, the notion of closeness is specified on the basis of comparing a-traces. However, as mentioned in the previous section, an a-trace is an abstraction of a trace of a HIOA, in which, the actions are not explicitly specified. As a result, a natural way for extending the notion of closeness to HIOA is to augment the closeness definition (Definition \ref{def:closeness}) with additional conditions that also compare the corresponding actions of the traces. 

We extend the definition of a-trace such that it can reflect the occurrence of  actions as well as the change of continuous variables. For this purpose, we consider each action as a special continuous variables. This variable is assumed to get only two values: $0$ and $\infty$, where a value of $0$ denotes the absence of the corresponding action while the value of $\infty$ denotes an instantaneous occurrence of the action. 

To this end, we first assume that the set of real numbers is augmented with a special number $\infty$ which is regarded to be larger than any number in $\mathbb{R}$, and hence larger than any arbitrary $\epsilon$. Further, for notational convenience,  we define $\infty-\infty=0$.
 Let $\mathbb{R}_\infty$ denote $\mathbb{R}$ extended with $\infty$. 
 
We extend the notion of trace and solution pair (Definitions \ref{def:trace} and \ref{def:sol}) such that for any action in a given HIOA, there is a  fresh variable. Such variables get only two values, namely 0 and $\infty$. This naturally extends the notion of a-trace to our notion of trace which has the same type (apart from the fact that $Val(V)$ ranges over  $\mathbb{R}_\infty$); the only difference is that the set of variables also includes the freshly introduced action variables.  

In order to extend the notion of closeness (and thus, hconf conformance relation), we need to define the meaning of a norm ($\| \|$) for the signals containing a mixture of normal real-valued variables and these extended type of variables.
Assume $y_1$ and $y_2$ are two vector of variables of the same length, i.e., 
\begin{eqnarray}
y_1=(y_1[1],\ldots, y_1[n])  \nonumber \\
y_2=(y_2[1],\ldots, y_2[n])  \nonumber 
\end{eqnarray}
We define the distance of $y_1$ and $y_2$, denoted by $\| y_2 - y_1 \|_e$, as
\begin{eqnarray}
\label{eq:normExt}
\| y_2 - y_1 \|_e = 
	\begin{cases}
		\infty,  & \exists i \in \{1, \ldots, n \} s{.}t{.} (y_1[i] = \infty \ and \ y_2[i] \neq \infty)\ or\ (y_1[i] \neq \infty \ and\ y_2[i] = \infty)\\
		\| y_2 - y_1 \|, & \text{otherwise}
	\end{cases}
\end{eqnarray}
where $\| y_2 - y_1 \|$ denotes the ordinary Euclidean distance.
Then, we can modify the definition of ($\tau$,$\epsilon$)-closeness based on this notion of distance.

{
\begin{definition}[Modified ($\tau$,$\epsilon$)-closeness]
\label{def:modifCloseness}
Consider a test duration $T \in \mathbb{R}_{+}$, a maximum number of jumps $J \in \mathbb{N}$, and $\tau,\epsilon > 0$; then two a-traces $y_1$ and $y_2$ are said to be  ($\tau$,$\epsilon$)-close, denoted by $y_1 \approx_{(\tau,\epsilon)}^e y_2$, if 
\begin{enumerate}
	\item
	for all $(t,i) \in \mathrm{dom}(y_1)$ with $t\leq T, i \leq J$, there exists $(s,j) \in \mathrm{dom}(y_2)$  such that  $|t-s| \leq \tau$ and $\|y_1(t,i) -y_2(s,j)\|_e \leq \epsilon$, and
	\item
	for all $(t,i) \in \mathrm{dom}(y_2)$ with $t\leq T, i \leq J$, there exists $(s,j) \in \mathrm{dom}(y_1)$  such that  $|t-s| \leq \tau$ and $\|y_2(t, i) -y_1(s, j)\|_e \leq \epsilon$.
\end{enumerate}
\end{definition}
}

{
\begin{definition}[Conformance Relation]
\label{def:modifConformance}
Consider two hybrid I/O automata $\mathscr{H}_1$ and $\mathscr{H}_2$. Given a test duration $T \in \mathbb{R}_{+}$, a maximum number of jumps $J \in \mathbb{N}$, and $\tau,\epsilon > 0$, \emph{$\mathscr{H}_2$ conforms to $\mathscr{H}_1$}, denoted by $\mathscr{H}_2 \approx_{(\tau,\epsilon)}^e \mathscr{H}_1$, if and only if for all solution pairs $(u,y_1)$ of $\mathscr{H}_1$,
 there exists a solution pair $(u, y_2)$ of $\mathscr{H}_2$ such that the corresponding output a-traces $y_1$ and $y_2$ are ($\tau$,$\epsilon$)-close, considering the modified notion of closeness in Definition \ref{def:modifCloseness}.
\end{definition}
}

\begin{example}
{Consider two traces $\alpha=\tau_1,\text{OFF},\tau_2$ and $\alpha'=\tau'_1,\tau'_2$ where $\tau_1$, $\tau_2$, $\tau'_1$, and $\tau'_2$ are trajectories shown in Fig. \ref{fig:hconf}. Let $\phi_1$ and $\phi_2$ be the corresponding a-traces obtained according to the abovementioned method for encoding the actions as numeric variables. Then, according to the modified notion of closeness, we will have $\phi_1 \not\approx_{(\tau,\epsilon)}^e \phi_2$.}
\end{example}

\subsection{Step-wise Approximate Refinement}
In order to apply our notion of conformance in a step-wise refinement trajectory, 
one needs to compose approximate conformance relations. 
This can be achieved through the following notion of semi-transitivity.

\begin{definition}[Semi-Transitivity] \label{def:tran}
Let $\mathscr{H}_1$, $\mathscr{H}_2$, and $\mathscr{H}_3$ be three arbitrary hybrid I/O automata such that $\mathscr{H}_2 \approx_{(\tau_1,\epsilon_1)} \mathscr{H}_1$ and $\mathscr{H}_3 \approx_{(\tau_2,\epsilon_2)} \mathscr{H}_2$ for some $\tau_1,\tau_2,\epsilon_1,\epsilon_2 > 0$. Then, the conformance relation is \emph{semi-transitive} if  $\mathscr{H}_3 \approx_{(\tau_1+\tau_2,\epsilon_1+\epsilon_2)} \mathscr{H}_1$.
\end{definition}

\begin{theorem} The extended conformance relation is semi-transitive.
\end{theorem}

\begin{proof}
For the proof, we first formally specify the assumptions and the thesis of the theorem; the latter is further refined and proven in a number of steps:

\begin{enumerate}
\item
Assumption: $\mathscr{H}_1 \approx_{(\tau_1,\epsilon_1)}^e \mathscr{H}_2$ and $\mathscr{H}_2 \approx_{(\tau_2,\epsilon_2)}^e \mathscr{H}_3$.
\\
Claim: $\mathscr{H}_1 \approx_{(\tau_1+\tau_2,\epsilon_1+\epsilon_2)}^e \mathscr{H}_3$.

\item
Assumption Rewriting (according to Definition \ref{def:modifCloseness}): 
\begin{eqnarray}
\forall (u,y_1) \in \mathscr{H}_1: \exists (u,y_2) \in \mathscr{H}_2 \text{ such  that } y_1 \approx_{(\tau_1,\epsilon_1)}^e y_2, \\
\forall (u,y_2) \in \mathscr{H}_2: \exists (u,y_3) \in \mathscr{H}_3 \text{ such that } y_2 \approx_{(\tau_2,\epsilon_2)}^e y_3.
\end{eqnarray}

Claim: 
$\forall (u,y_1) \in \mathscr{H}_1: \exists (u,y_3) \in \mathscr{H}_3 \text{ such  that } y_1 \approx_{(\tau_1+\tau_2,\epsilon_1+\epsilon_2)}^e y_3$.

\item

Assumption Rewriting: 
\begin{multline}
\label{eq:firstAssum}
\forall (u,y_1) \in \mathscr{H}_1: \exists (u,y_2) \in \mathscr{H}_2 \text{ such  that } \\
 \forall(t,i) \in \mathrm{dom}(y_1): \exists(s,j) \in \mathrm{dom}(y_2) \text{ such that } |t-s| \leq \tau_1 \text{ and } \|y_1(t,i) -y_2(s,j)\|_e \leq \epsilon_1, 
\end{multline}
and also
\begin{multline}
\label{eq:secondAssum}
\forall (u,y_2) \in \mathscr{H}_2: \exists (u,y_3) \in \mathscr{H}_3 \text{ such  that } \\
 \forall(s,j) \in \mathrm{dom}(y_2): \exists(v,k) \in \mathrm{dom}(y_3) \text{ such that } |v-s| \leq \tau_2 \text{ and } \|y_2(s,j) -y_3(v,k)\|_e \leq \epsilon_2, 
\end{multline}

conditioned on that $t\leq T$ and $i \leq J$.

Claim: 
\begin{multline}
\label{eq:claim1}
\forall (u,y_1) \in \mathscr{H}_1: \exists (u,y_3) \in \mathscr{H}_3 \text{ such  that }  \\
 \forall(t,i) \in \mathrm{dom}(y_1): \exists(v,k) \in \mathrm{dom}(y_3) \text{ such that } |v-t| \leq \tau_1+\tau_2 \text{ and } \|y_1(t,i) -y_3(v,k)\|_e \leq \epsilon_1+\epsilon_2.
\end{multline}
\end{enumerate}

Now, we proceed with  the proof.
Fix an arbitrary $(u,y_1) \in \mathscr{H}_1$; it follows from \eqref{eq:firstAssum} that  $\exists y_2$ such that $(u,y_2) \in \mathscr{H}_2$ for which $\forall(t,i) \in \mathrm{dom}(y_1)$  there exists $(s,j) \in \mathrm{dom}(y_2)$ such that
\begin{equation}
\label{eq:y12timeCloseness}
 |t-s| \leq \tau_1 
\end{equation}
and
\begin{equation}
\label{eq:y12closeness}
\|y_1(t,i) - y_2(s,j)\|_e \leq \epsilon_1
\end{equation}
Similarly, from \eqref{eq:secondAssum}, it follows that considering $y_2$, $\exists y_3$ such that $(u,y_3) \in \mathscr{H}_3$ for which $\forall(s,j) \in \mathrm{dom}(y_2)$  there exists $(v,k) \in \mathrm{dom}(y_3)$ such that
\begin{equation}
\label{eq:y23timeCloseness}
 |s-v| \leq \tau_2
\end{equation}
and
\begin{equation}
\label{eq:y23closeness}
\|y_2(s,j) - y_3(v,k)\|_e \leq \epsilon_2.
\end{equation}

Now, fix an arbitrary $\forall(t,i) \in \mathrm{dom}(y_1)$ and assume $(s,j)$ and $(v,k)$ are the corresponding points as specified above.
Since $\epsilon_1$ has a bounded value (i.e., $\epsilon_1 < \infty$), we can conclude from \eqref{eq:y12closeness} that the value of those variables in $y_1(t,i)$ and $y_2(s',j')$ which are associated with the actions is exactly the same. This is because that otherwise $\|y_1(t,i) - y_2(s',j')\|_e = \infty > \epsilon_1$ (according to \eqref{eq:normExt}) which violates \eqref{eq:y12closeness}. 
The same results hold for $y_2(s',j')$ with respect to $y_3(v,k)$.
As a consequence, the norm operator in \eqref{eq:y12closeness} and  \eqref{eq:y23closeness} is reduced to the normal Euclidean norm which yields the following relation (according to the triangular property of the Euclidean norm)
\begin{equation}
\|y_1(t,i) - y_3(v,k)\|_e \leq \epsilon_1 + \epsilon_2.
\end{equation}
Besides, from \eqref{eq:y12timeCloseness} and \eqref{eq:y23timeCloseness} it is concluded that 
\begin{equation}
 |t-v| \leq \tau_1 + \tau_2
\end{equation}
Consequently, the claim declared in \eqref{eq:claim1} is achieved.
\end{proof}

\section{Conclusions and Discussion}\label{sec:conc}
In this paper, we proposed a notion of conformance parameterized by conformance bounds in time and value. In addition to approximate comparison of (sampled) continuous signals, our notion allows for approximate matching of discrete signals by allowing them to happen within specified time margins. In this sense, it consolidates two earlier notions of hybrid conformance by Abbas and Fainekos \cite{Abbas14}  and by van Osch \cite{vanOsch2006}.

In the remainder of this section, we present the directions of our ongoing research. 
As an immediate next step, we would like to prove the conservative extension property of the extended conformance relation compared to the other conformance relations, namely hioco and hconf. This will generalize our earlier result in \cite{Mohaqeqi2014}.
Moreover, we envisage other important issues such as sound sampling rates \cite{Mohaqeqi2016}, test-case generation algorithms, coverage \cite{Dreossi2015,Fainekos2015}, links to temporal and modal logics \cite{Deshmukh15}, and compositionality \cite{Abbas2015,Benes2015}.



\bibliographystyle{eptcs}
\bibliography{generic}
\end{document}